\documentclass[11pt]{article}
\usepackage{amsmath,amsthm}
\usepackage{amssymb}
\usepackage[dvipdfmx]{graphicx,color}
\allowdisplaybreaks
\nopagebreak[3]
\newcommand{\newsection}[1]
{\section{#1}\setcounter{theorem}{0} \setcounter{equation}{0} \par\noindent}

\newtheorem{theorem}{Theorem}

\newtheorem{lemma}[theorem]{Lemma}
\newtheorem{corollary}[theorem]{Corollary}
\newtheorem{proposition}[theorem]{Proposition}

\newtheorem{remark}[theorem]{Remark}

\newcommand{\beq}{ \begin{equation} }
\newcommand{\eeq}{ \end{equation} }

\newcommand{\br}{{\mathbb R}}
\newcommand{\bc}{{\mathbb C}}

\newcommand{\supp}{\mbox{\rm supp}\ }

\newcommand{\brn}{ { \mathbb{R}^n } }
\renewcommand{\Re}{\operatorname{Re}}

\newcommand{\four}{I\hspace{-1mm}V}
\newcommand{\six}{V\hspace{-1mm}I}
\newcommand{\seven}{V\hspace{-1mm}I\hspace{-1mm}I}
\newcommand{\eight}{V\hspace{-1mm}I\hspace{-1mm}I\hspace{-1mm}I}

\title{
Blowing-up solutions of  Klein-Gordon equations 
\\ 
with gauge variant semilinear terms 
\\
in Friedmann-Lema\'itre-Robertson-Walker spacetimes 
under finite speed of propagation
}
\author
{
Makoto NAKAMURA
\thanks
{Graduate School of Information Science and Technology,  
Osaka University, 
1-5 Yamadaoka, Suita, Osaka 565-0871, JAPAN.  
E-mail:  \texttt{makoto.nakamura.ist@osaka-u.ac.jp} }
\ \ and \ 
Takuma YOSHIZUMI
\thanks
{Graduate School of Information Science and Technology,  
Osaka University, 
1-5 Yamadaoka, Suita, Osaka 565-0871, JAPAN.  
E-mail:  \texttt{yoshizumi.takuma@ist.osaka-u.ac.jp} }
}

\date{}

\begin{document}

\maketitle

\begin{abstract}
Blowing-up solutions of Klein-Gordon equations with gauge variant semilinear terms are considered in Friedmann-Lema\'itre-Robertson-Walker spacetimes. 
Effects of spatial expansion or contraction on the solutions are studied through the scale-function and the curved mass.
\end{abstract}


\vspace{5pt}


%

\section{Introduction}

We show some blowing-up solutions of Klein-Gordon equations with gauge-variant semilinear terms in Friedmann-Lema\'itre-Robertson-Walker spacetimes (FLRW spacetimes for short).
FLRW spacetimes are solutions of the Einstein equations 
with the cosmological constant under the cosmological principle.
They describe the spatial expansion or contraction, 
and yield some important models of the universe.
Let $n\ge1$ be the spatial dimension, 
$a(\cdot)>0$ be a scale-function defined on an interval $[0,T_0)$ for some $0<T_0\le\infty$,
and let $c>0$ be the speed of light.
The metrics $\{g_{\alpha\beta}\}$ of FLRW spacetimes are expressed by 
$\sum_{0\le \alpha,\beta\le n}g_{\alpha\beta}dx^\alpha dx^\beta
:=
-c^2(dt)^2+
a^2(t)
\sum_{j=1}^n (dx^j)^2$,
where we have put the spatial curvature as zero,  
and $x^0=t$ is the time-variable 
(see e.g., 
\cite{Carroll-2004-Addison, DInverno-1992-Oxford}).
When $a$ is a positive constant, 
the spacetime 
reduces to the Minkowski spacetime.

We denote the first and second derivatives of one variable function $a$ by $\dot{a}$ and $\ddot{a}$.
The Klein-Gordon equation generated by the above metric 
$(g_{\alpha\beta})_{0\le \alpha,\beta\le n}$ is given by  
$(\sqrt{|g|})^{-1}\partial_\alpha\left(\sqrt{|g|}g^{\alpha\beta}\partial_\beta v\right)=m^2v+f(v)$ 
for the determinant $g:=\mbox{det}\left(g_{\alpha\beta}\right)$ 
and the inverse matrix $\left(g^{\alpha\beta}\right)$, 
i.e.,
\beq
\label{Cauchy-0}
\frac{1}{c^2}\partial_t^2v+\frac{n\dot{a}}{c^2a}\partial_tv-\frac{1}{a^2} \Delta v +m^2v+f(v)
=0,
\eeq
where 
$m$ denotes the mass,  
$\Delta:=\sum_{j=1}^n \partial_j^2$ denotes the Laplacian, 
and $f$ denotes a force term.
This equation is rewritten as 
\beq
\label{Eq-KGLambda}
c^{-2}\partial_t^2u-a^{-2}\Delta u
+M^2u
+a^{n/2}f(a^{-n/2}u)=0
\eeq
by the transformation $u:=va^{n/2}$,
where the function $M$ is called curved mass defined by  
\beq
\label{Def-M}
M^2
=
m^2-\frac{n(n-2)}{4c^2}\left( \frac{\dot{a}}{a} \right)^2
-\frac{n}{2c^2}\cdot\frac{\ddot{a}}{a},
\ \ \ \ 
M:=\sqrt {M^2}.
\eeq

In this paper, we consider the Cauchy problem in the gauge-variant case $f(v)=\lambda |v|^p$ in 
\eqref{Eq-KGLambda}, namely,    
\beq
\label{Cauchy}
\left\{
\begin{array}{l}
(c^{-2}\partial_t^2-a^{-2}(t) \Delta+M^2(t))u(t,x)-\lambda a^{-n(p-1)/2}(t)|u|^p(t,x)=0,
\\
u(0,\cdot)=u_0(\cdot),\ \ \partial_tu(0,\cdot)=u_1(\cdot)
\end{array}
\right.
\eeq
for $(t,x)\in [0,T)\times\br^n$ and $0<T\le T_0$, 
where 
$u_0$ and $u_1$ are given initial data.
We expect some dissipative effects when $\dot{a}>0$ in \eqref{Cauchy-0} 
from the term $n\dot{a}\partial_t v/c^2a$.
Actually, the spatial expansion yields dissipative effects in energy estimates for  the Klein-Gordon equation 
(see 
\cite{Nakamura-2014-JMAA, Nakamura-Yoshizumi-9999}).

Blowing-up solutions for \eqref{Cauchy} have been considered in 
\cite[Proposition 1]{Nakamura-2021-JMP} and 
\cite[Theorems 1.1 and 1.2]{Yagdjian-2009-DCDS} 
in the de Sitter spacetime.
The aim of this paper is to extend these results into general FLRW spacetimes including concrete examples of scale functions of Big-Rip and Big-Crunch.
We consider the real and purely imaginary mass $m\in \br\cup i\br$ 
since $m\in i\br$ plays an important role when we consider the breakdown of symmetry of potential well, or in $\phi^4$-theory in physics 
(see \cite{Nakamura-2021-JMP}).
We refer to \cite{Wei-Yong-2024-JMP} and the references therein 
on blowing-up solutions for small initial data in FLRW spacetimes 
for gauge-variant semilinear terms and vanishing mass $m=0$, i.e., wave equation, 
under some conditions on $p$ concerned with the Strauss conjecture 
(see also \cite{Galstian-Yagdjian-2020-RMP} in the Einstein-de Sitter spacetime),
while our results are for relatively large data and arbitrary $m\in \br\cup i\br$ and $1<p<\infty$.

The case of the gauge invariant semilinear term of the form 
$f(u)=\lambda |v|^{p-1}v$ in \eqref{Eq-KGLambda} 
has been considered by McCollum, Mwamba and Oliver in 
\cite[Theorem 1]{McCollum-Mwamba-Oliver-2024-NA} for $\dot{a}>0$, 
and in \cite[Theorem 1.4]{Nakamura-Yoshizumi-9999} for $\dot{a}\le 0$ by different techniques from this paper to show the blowing-up of the norm $\|u\|_2$,
while the blowing-up of $\int_\brn u(t,x)dx$ is shown in this paper.

The Cauchy problem of \eqref{Cauchy} including gauge invariant semilinear term is widely studied especially on local and global solutions, and their asymptotic behaviors 
(see 
\cite{Baskin-2012-AHP,
Galstian-Yagdjian-2015-NA-TMA, 
Hintz-Vasy-2015-AnalysisPDE,
Nakamura-2021-JMP,
Nakamura-Yoshizumi-9999,
Yagdjian-2012-JMAA} for closely related results).
We are focused on blowing-up solutions in this paper.

\vspace{9pt}


Now, we introduce some concrete examples of the scale-function $a$ and the curved mass $M$.
For $\sigma\in\mathbb{R}$ and the Hubble constant $H\in \br$, 
we put
\begin{equation}
\label{R-Def-T_0}
T_0:= 
\begin{cases}
\infty & \mbox{if}\ \ (1+\sigma)H\ge0,
\\
-\frac{2}{n(1+\sigma)H}(>0) & \mbox{if}\ \ (1+\sigma)H<0,
\end{cases}
\end{equation}
and define $a(\cdot)$ by 
\begin{equation}
\label{Def-a}
a(t):=
\begin{cases}
a_0\left\{1+\frac{n(1+\sigma)Ht}{2}\right\}^{2/n(1+\sigma)} & \mbox{if}\ \ \sigma\ne-1,
\\
a_0\exp(Ht) & \mbox{if}\ \ \sigma=-1
\end{cases}
\end{equation}
for $0\le t<T_0$. 
We note $a_0=a(0)$ and $H=\dot{a}(0)/a(0)$.
This scale-function $a(\cdot)$ describes the Minkowski space when $H=0$ 
(namely, $a(\cdot)$ is the constant $a_0$),
the expanding space when $H>0$ with $\sigma\ge -1$,
the blowing-up space when $H>0$ with $\sigma< -1$ 
(the ``Big-Rip'' in cosmology), 
the contracting space when $H<0$ with $\sigma\le -1$, 
and the vanishing space when $H<0$ with $\sigma> -1$ 
(the ``Big-Crunch'' in cosmology).
It describes the de Sitter spacetime when $\sigma=-1$.
The corresponding curved mass $M$ defined by \eqref{Def-M} is rewritten as  
\beq
\label{M-FLRW}
M^2
=
m^2+\sigma\left(\frac{nH}{2c}\right)^2\cdot 
\left\{1+\frac{n(1+\sigma)Ht}{2}\right\}^{-2}
\eeq
(see (2) in Lemma \ref{Lem-11}, below).

We say that the solution $u$ of \eqref{Cauchy} is a global solution if $u$ exists on the time-interval $[0,T_0)$ 
since $T_0$ is the end of the spacetime.
We note that the squared curved mass $M^2$ changes its signature at $T_1$ 
when $(1+\sigma)H<0$, $\sigma<0$, $m>\sqrt{|\sigma|}n|H|/2c$,
where 
\beq
\label{Def-T1}
T_1:=
-\frac{2}{n(1+\sigma)H}\left( 1-\frac{ \sqrt{|\sigma|}n|H| }{2cm} \right). 
\eeq
We also consider the case $M^2< 0$, 
i.e., purely imaginary curved mass $M\in i\br$,  
since it naturally appears in FLRW spacetimes.

For any function $a\in C^1([0,T_0),(0,\infty))$ and $0<r_0<\infty$, 
define $r(\cdot)$ and $q(\cdot)$ by  
\beq
\label{Def-r-q}
r(t):=r_0+\int_0^t \frac{c}{a(s)} ds, 
\ \ \ \ 
q(t):=\frac{a(t)r^2(t)}{a_0}.
\eeq
for $0\le t<T_0$.
Assume 
$q$ be a monotone function, and put 
\beq
\label{Def-TildeQ}
\widetilde{q}(t):=
\begin{cases}
q_0 & \mbox{if}\ \ \dot{q}\le 0,
\\
q(t) & \mbox{if}\ \ \dot{q}\ge 0
\end{cases}
\eeq
for $0\le t<T_0$,
where $q_0:=q(0)$ and $\dot{q}:=dq/dt$.

We denote the Lebesgue space by $L^q(I)$ for an interval $I\subset \br$ 
and $1\le q\le \infty$ with the norm 
\[
\|Y\|_{L^q(I)}:=
\begin{cases}
\left\{
\int_I |Y(t)|^q dt
\right\}^{1/q} & \mbox{if}\ \ 1\le q<\infty,
\\
{\mbox{ess.}\sup}_{t\in I} |Y(t)| 
& \mbox{if}\ \ q=\infty.
\end{cases}
\]

We obtain the following result on blowing-up solutions of the Cauchy problem \eqref{Cauchy}.

\begin{theorem}[Blowing-up solutions]
\label{Thm-22}
Let $T_0$, $a$ and $M$ satisfy 
\beq
\label{Condition-aM}
0<T_0\le \infty,
\ \ 
a\in C^2([0,T_0),(0,\infty)),
\ \ 
M^2\in C([0,T_0),\br). 
\eeq
Let 
$n\ge1$, 
$\lambda>0$, 
$1<p<\infty$, 
$0<\varepsilon<1$,
$0<\theta<1$,
$r_0>0$.
Let $\supp u_0\cup  \supp u_1\subset \{x\in \brn;\ |x|\le r_0\}$.
Let $N$ be any real number which satisfies  
\beq
\label{Cond-NM}
N\ge0, \ \ \ \ N^2+\inf_{0<t<T_0} M^2(t)\ge0.
\eeq
Let $r$ and $q$ be defined by \eqref{Def-r-q}.
Assume $q$ is a monotone function, and $\widetilde{q}$ defined by 
\eqref{Def-TildeQ} satisfies 
\beq
\label{Assume-AB}
A:=\inf_{0<t<T_0}  \frac{e^{cN(1-\varepsilon)t} }{ \widetilde{q}^{n/2}(t) }>0,
\ \ 
B:=\sup_{0<t<T_0} 
\frac{\widetilde{q}^{n/2}(t) \left\{N^2+M^2(t)\right\}^{1/(p-1)} }{e^{cNt} }<\infty.
\eeq
Put 
\beq
\label{Def-w0-w1}
w_0:=\Re\int_\brn u_0(x)dx,\ \ 
w_1:=\Re\int_\brn u_1(x)dx.
\eeq
Assume $w_0$ and $w_1$ satisfy 
\beq
\label{Assume-w0}
w_0>\frac{Q^{n/2}B}{\{(1-\theta)\lambda\}^{1/(p-1)} },
\eeq
\beq
\label{Assume-w1}
w_1\ge 
\max\left\{
\ cNw_0,
\ \ 
\sqrt{\frac{2\lambda c^2\theta}{p+1} }\cdot
\frac{w_0^{(p+1)/2}}{ (r_0^2Q)^{n(p-1)/4} }
\ \right\},
\eeq
where $Q:=\omega_n^{2/n} a_0$ 
and $\omega_n$ denotes the volume of unit ball in $\brn$.
Put 
\beq
\label{Def-D-T}
D:=\frac{2c^2\theta \lambda A^{p-1}}{(p+1)Q^{n(p-1)/2}},
\ \ \ \ 
T_\ast:=\frac{2}{\varepsilon (p-1) \sqrt{D}  w_0^{(p-1)/2} }.
\eeq

Then the solution $u\in C^2([0,T)\times \brn)$ of the Cauchy problem \eqref{Cauchy} blows-up in finite time in $L^1(\brn)$ 
no later than $T_\ast$ if $T_\ast \le T_0$.
More precisely, 
$\Re\int_\brn u(t,x)dx$ $\to\infty$ as $t\nearrow T$ for some $T$ with $0<T\le T_\ast$.
\end{theorem}

\vspace{10pt}

We note our blowing-up solutions are obtained for relatively large data by the assumptions \eqref{Assume-w0} and \eqref{Assume-w1}. 
Our result shows that the upper bound of the life-span given by $T_\ast$ in \eqref{Def-D-T} has the order $T_\ast=O(w_0^{-(p-1)/2})$.
We obtain the following corollary from Theorem \ref{Thm-22} 
for the scale function and the curved mass in  \eqref{Def-a} and \eqref{M-FLRW}.

\begin{corollary}[Blowing-up solutions under \eqref{Def-a} and \eqref{M-FLRW}]
\label{Cor-24}
Let 
$n\ge1$, 
$\lambda>0$,
$1<p<\infty$, 
$0<\varepsilon<1$,
$0<\theta<1$,
$r_0>0$.
Let $m\in \br\cup i\br$.
Let $H\in \br$, $\sigma\in \br$ and $N>0$ satisfy one of the following conditions from (i) to (viii).

\vspace{3pt}

(i) $H=0$, $\sigma\in \br$, $N^2+m^2\ge0$, $r_0>0$.

(ii) $H>0$, $\sigma\ge0$, $N^2+m^2\ge0$, $r_0>0$.

(iii) $H>0$, $-1<\sigma<0$, $N^2+m^2+\sigma\left(\frac{nH}{2c}\right)^2\ge0$, $r_0>0$.

(iv) $H>0$, $\sigma=-1$, $N^2+m^2-\left(\frac{nH}{2c}\right)^2\ge0$, $N>\frac{nH}{2c(1-\varepsilon)}$, $r_0>0$.

(v) $H<0$, $\sigma>0$, $N^2+m^2+\sigma\left(\frac{nH}{2c}\right)^2\ge0$, $r_0\ge-\frac{2c}{a_0H}$.

(vi) $H<0$, $\sigma=0$, $N^2+m^2\ge0$, $r_0>0$.
Moreover, $r_0\ge -\frac{2c}{a_0H}$ if $n\ge2$.

(vii) $H<0$, $\sigma=-1$, $N^2+m^2-\left(\frac{nH}{2c}\right)^2\ge0$, $r_0\le \frac{2c}{a_0|H|}$, $N>\frac{n|H|}{2c(1-\varepsilon)}$.

(viii) $H<0$, $\sigma<-1$, $N^2+m^2+\sigma\left(\frac{nH}{2c}\right)^2\ge0$, $r_0\le \frac{2c}{a_0|H|}$.

\vspace{3pt}

Moreover, for $T_0$ given by \eqref{R-Def-T_0}, 
let $a$ and $M$ be the scale-function and the curved mass given by \eqref{Def-a} and \eqref{M-FLRW}.
Let $\supp u_0\cup  \supp u_1\subset \{x\in \brn;\ |x|\le r_0\}$.
Let $w_0$ and $w_1$ be the numbers defined by \eqref{Def-w0-w1} 
with \eqref{Assume-w0} and \eqref{Assume-w1}.
Let $D$ and $T_\ast$ be numbers defined by \eqref{Def-D-T}.

Then 
the condition \eqref{Cond-NM} holds, 
$q$ defined by \eqref{Def-r-q} is a monotone function, 
and $\widetilde{q}$ defined by \eqref{Def-TildeQ} satisfies \eqref{Assume-AB}.
So that, 
the all assumptions in Theorem \ref{Thm-22} are satisfied, and 
the result in Theorem \ref{Thm-22} holds.
Namely, the solution $u\in C^2([0,T)\times \brn)$ of the Cauchy problem \eqref{Cauchy} blows-up in finite time in $L^1(\brn)$ 
no later than $T_\ast$ if $T_\ast \le T_0$.
More precisely, 
$\Re\int_\brn u(t,x)dx$ $\to\infty$ as $t\nearrow T$ for some $T$ with $0<T\le T_\ast$.
\end{corollary}

\begin{remark}
The condition $r_0\ge-2c/a_0 H$ in (v) and (vi) in Corollary \ref{Cor-24} is from the monotonicity assumption on $\dot{q}\le 0$ in Theorem \ref{Thm-22}.
Under this condition, $\widetilde{q}$ in \eqref{Def-TildeQ} is constant, 
and the second assumption in \eqref{Assume-w1} becomes weaker as $r_0$ tends to $\infty$, while \eqref{Assume-w0} remains invariant.
So that, blowing-up solutions are obtained for wider $w_1$ as the support of initial data tends to large.
\end{remark}

Kato \cite{Kato-1980-CPAM} has shown the blowing-up of $\int_\brn u(t,x) dx$ in finite time 
in the case of $H=0$ and a general elliptic operator instead of $\Delta$.
Yagdjian \cite[Theorems 1.1 and 1.2]{Yagdjian-2009-DCDS} has considered the case $H=1$ with $M^2<0$.
His method is also valid in the case $m\in i\br$ which is applied to the Klein-Gordon equation with the Higgs potential 
(see the introduction in \cite[(8)]{Balogh-Banda-Yagdjian-2019-CommNonlinearSciNumerSimulat}, 
and \cite[Theorem 3.1]{Yagdjian-2012-JMAA}).
The case $H\in \br$ and $\sigma=-1$ has been considered in 
\cite[Proposition 1]{Nakamura-2021-JMP} under $m\in i\br$ and $H>-c/r_0$.
We extend these results into general case $H\in \br$, $\sigma\in \br$ and $m\in \br\cup i\br$, 
and we show how the spatial variance affects behaviors of solutions.

The excluded case of the conditions from (i) to (viii) on $H$ and $\sigma$ 
in Corollary \ref{Cor-24} is the case $\sigma<0$ with $(1+\sigma)H<0$, 
under which the curved mass $M^2$ tends to $-\infty$ as $t$ tends to $T_0$ 
(see (vi) in Lemma \ref{Lem-11}, below).
Our argument is not valid for this case, which requires further techniques.

We denote the inequality $A\le CB$ by $A\lesssim B$ for some constant $C>0$ which is not essential for the argument.

This paper is organized as follows.
In Section \ref{Sec-Pre}, 
we collect fundamental properties on the scale function and the curved mass, 
some ordinary differential equation of second order, 
and some auxiliary functions corresponding to the finite speed of propagation of waves in FLRW spacetimes, 
which are required to prove 
Theorem \ref{Thm-22} and  
Corollary \ref{Cor-24} 
in Sections \ref{Sec-Thm-22} and \ref{Sec-Cor-24}.
An appendix is included on the finite speed of propagation in Section \ref{Sec-Appendix}.


\newsection{Preliminaries}
\label{Sec-Pre}
We use the following properties of $a$ and $M$ in the case of FLRW  spacetimes, 
which follow from direct calculation, to prove the results in the previous section.

\begin{lemma}
\label{Lem-11}
(\cite[Lemmas 2.6 and 2.7]{Nakamura-Yoshizumi-9999}.)
Let $m\in \br$, $H\in \br$ and $\sigma\in \br$.
Let $T_0$, $a$ and $M$ be given by \eqref{R-Def-T_0}, \eqref{Def-a} and \eqref{M-FLRW}.
Then the following results hold.

(1) 
\ $\frac{\dot{a}}{{a}}=H\left(\frac{a}{a_0}\right)^{-n(1+\sigma)/2}$.

(2) 
\ $\displaystyle M^2
=
m^2+\sigma\left(\frac{nH}{2c}\right)^2
\left\{1+\frac{n(1+\sigma)Ht}{2}\right\}^{-2}$.

(3) 
\ $M^2$ satisfies 
\[
\begin{array}{ll}
(i) \ \ M^2=m^2 & \mbox{if }\ H=0, \ \mbox{or}\ \sigma=0,
\\
(ii)\ \ M^2=m^2-\left(\frac{nH}{2c}\right)^2 & \mbox{if }\ H\neq0, \ \sigma=-1,
\\
(iii)\ \ m^2<M^2\le m^2+\sigma\left(\frac{nH}{2c}\right)^2 & \mbox{if }\ H>0, \ \sigma>0,
\\
(iv)\ \ m^2+\sigma\left(\frac{nH}{2c}\right)^2\le M^2<m^2 & \mbox{if }\ (1+\sigma)H>0, \ \sigma<0,
\\
(v)\ \ m^2+\sigma\left(\frac{nH}{2c}\right)^2\le M^2\to \infty \ (t\to T_0) & \mbox{if }\ H<0, \ \sigma>0,
\\
(vi)\ \ m^2+\sigma\left(\frac{nH}{2c}\right)^2\ge M^2\to -\infty \ (t\to T_0) & \mbox{if }\ (1+\sigma)H<0, \ \sigma<0.
\end{array}
\]
\end{lemma}

We prepare the following fundamental result on the positivity of the solution of an ordinary differential equation of second order.

\begin{lemma}
\label{Lem-21}
Let $M^2\in C([0,T_0),\br)$, $b\in C([0,T_0),(0,\infty))$.
Let $0<\theta<1$, $1<p<\infty$
Let $N$ be a number with $0\le N<\infty$, $N^2+\inf_{0<t<T_0} M^2\ge0$.
Let $w_0$ and $w_1$ be real numbers with 
\beq
\label{w0}
w_0>\sup_{0<t<T_0} 
\left[
e^{-cNt} 
\left\{
\frac{N^2+M^2(t)}{(1-\theta)b(t)}
\right\}^{1/(p-1)}
\right],
\ \ 
w_1\ge cN w_0.
\eeq
Then the solution $w$ of the Cauchy problem 
\beq
\label{Cauchy-w}
\begin{cases}
c^{-2} \ddot{w}(t)+M^2(t)w(t)-b(t)|w|^p(t)\ge 0, 
\\
w(0)=w_0,\ \ \dot{w}(0)=w_1
\end{cases}
\eeq
for $0\le t<T_0$ satisfies the followings, 
where $\dot{w}:=dw/dt$ and $\ddot{w}:=d^2w/dt^2$.

(1) $w(t)\ge w_0 e^{cNt}$.

(2) $(1-\theta)b w^{p-1}-M^2>N^2$.

(3) $c^{-2} \ddot{w}-N^2w-\theta bw^p\ge0$.

(4) $\dot{w}\ge w_1$.
\end{lemma}

\begin{proof}
We assume that $t_0$ defined by 
\[
t_0:=\inf\{0\le t<T_0;\ w(t)\le 0\ \mbox{or}\ (1-\theta)b(t)|w|^{p-1}(t)-M^2(t)\le N^2\}
\]
exists and satisfies $0\le t_0<T_0$,
and we show a contradiction.
Since we have $w_0>0$ and $(1-\theta)b_0w_0^{p-1}-M_0^2>N^2$ 
by the first assumption in \eqref{w0}, 
we have $t_0\neq0$.
Since we have 
\[
w(t)>0,
\ \ 
(1-\theta)b(t)w^{p-1}(t)-M^2(t)> N^2
\]
for $0\le t<t_0$, 
we have 
\[
b w^p=(1-\theta)b w^p+\theta b w^p\ge (1-\theta)b w^p
\]
and 
\begin{eqnarray}
0 
&\le& 
c^{-2} \ddot{w}+M^2w-b|w|^p \nonumber \\
&=& 
c^{-2} \ddot{w}-
\left\{
(1-\theta)b w^{p-1}-M^2
\right\}w 
-\theta b w^p 
\nonumber \\
&\le& 
c^{-2} \ddot{w}-N^2w-\theta b w^p 
\nonumber\\
&\le& 
c^{-2} \ddot{w}-N^2w
\label{Proof-Lem-21-1000}
\end{eqnarray}
by the differential inequality in \eqref{Cauchy-w}.
Thus, we obtain 
\beq
\label{Proof-Lem-21-1300}
\dot{w}\ge w_1\ge0
\eeq
by $N^2w\ge0$,
where the second inequality follows from the second assumption in \eqref{w0}.
Multiplying $\dot{w}$ to the inequality $c^{-2}\ddot{w}-N^2w\ge0$ 
in \eqref{Proof-Lem-21-1000}, 
we have 
\[
c^{-2}\dot{w}^2-N^2w^2\ge c^{-2}w_1^2-N^2w_0^2\ge0,
\]
where the second inequality follows from the second assumption in \eqref{w0}.
So that, we obtain $\dot{w}\ge c N w$ by $w>0$ and $\dot{w}\ge0$, 
which yields 
\beq
\label{Proof-Lem-21-1500}
w(t)\ge w_0 e^{c N t}
\eeq 
for $0\le t<t_0$. 
Especially, we have 
\beq
\label{Proof-Lem-21-2000}
w(t_0)\ge w_0 e^{c N t_0}>0
\eeq 
as its limit.
By this and the first assumption in \eqref{w0}, we have 
\beq
\label{Proof-Lem-21-3000}
(1-\theta)b(t_0)|w|^{p-1}(t_0)-M^2(t_0)
\ge
(1-\theta)b(t_0)(w_0e^{cNt_0})^{p-1}-M^2(t_0)
> N^2.
\eeq
The results \eqref{Proof-Lem-21-2000} and \eqref{Proof-Lem-21-3000} 
contradicts to the definition of $t_0$.

Therefore, we obtain $w>0$ and $(1-\theta)bw^{p-1}-M^2>N^2$ on the interval $[0,T_0)$, which is the required result (2).
Repeating the above argument, 
the results (1), (3) and (4) follow from 
\eqref{Proof-Lem-21-1500},
\eqref{Proof-Lem-21-1000} 
and 
\eqref{Proof-Lem-21-1300}, 
respectively.
\end{proof}

We prepare the following result on the monotonicity of the function $q$ defined in \eqref{Def-r-q}.

\begin{lemma}
\label{Lem-23}
Let $H\in \br$, $\sigma\in \br$, $r_0>0$.
Let $a$, $r$, $q$ be defined by \eqref{Def-a} and \eqref{Def-r-q}.
Put 
\[
d:=r+\frac{2c}{a_0H} 
\left(
\frac{a}{a_0}
\right)^{n(1+\sigma)/2-1}
\]
for $H\neq0$.
Then the following results hold on $[0,T_0)$.

(1) $\dot{q}=\frac{2cr}{a_0}$ when $H=0$, 
and  
$\dot{q}=Hrd
\left(
\frac{a}{a_0}
\right)^{1-n(1+\sigma)/2}$ when $H\neq0$.

(2) 
$\dot{q}\ge0$ holds if $H\ge0$, $\sigma\in \br$, $r_0>0$, 
or if $H<0$, $\sigma\le-1+\frac{1}{n}$, $r_0\le -\frac{2c}{a_0H}$.

(3)
$\dot{q}\le0$ holds if $H<0$, $\sigma\ge-1+\frac{1}{n}$, $r_0\ge -\frac{2c}{a_0H}$.
\end{lemma}

\begin{proof}
(1) 
Since we have $\dot{a}/a=H(a/a_0)^{-n(1+\sigma)/2}$ and $\dot{r}=c/a$ by (1) in Lemma \ref{Lem-11} and the definition of $r$,
we obtain 
\beq
\label{Proof-Lem-23-1000}
\dot{q}=
\frac{\dot{a}r^2}{a_0}+\frac{2ar\dot{r}}{a_0}
=
Hr^2\left(\frac{a}{a_0}\right)^{1-n(1+\sigma)/2} +\frac{2cr}{a_0}
=
H r d
\left(
\frac{a}{a_0}
\right)^{1-n(1+\sigma)/2}
\eeq
as required, 
where the last equality holds for $H\neq0$.

We prove (2) and (3). 
We have 
\[
r=r_0+
\begin{cases}
\frac{ct}{a_0} & \mbox{if}\ H=0,\ \sigma \in \br,
\\
\frac{2c}{ a_0H\left\{2-n(1+\sigma)\right\} }
\left\{
1-
\left(
\frac{a}{a_0}
\right)^{n(1+\sigma)/2-1}
\right\}
& \mbox{if}\ H\neq0,\ \sigma \neq -1+\frac{2}{n},
\\
\frac{c\log(1+Ht)}{a_0H} 
& \mbox{if}\ H\neq0,\ \sigma = -1+\frac{2}{n}.
\end{cases}
\]

(i) We have $\dot{q}\ge0$ if $H\ge0$ by (1).
Thus, we consider the case $H<0$ in the following.

(ii) 
If $H<0$ and $\sigma\neq-1+2/n$, then we have 
\beq
\label{Proof-Lem-23-2000}
d=r_0-\frac{2c}{a_0H\left\{n(1+\sigma)-2\right\} }
\left[
1+\{1-n(1+\sigma)\}
\left(
\frac{a}{a_0}
\right)^{n(1+\sigma)/2-1}
\right].
\eeq
Especially, we have $d(0)=r_0+2c/a_0H$.

(iii) 
If $H<0$ and $\sigma<-1+1/n$, then we have 
$d(t)\searrow -\infty$ as $t\to T_0$ by \eqref{Proof-Lem-23-2000}.
Especially, we have $d\le 0$ if $r_0\le -2c/a_0H$.

(iv) 
If $H<0$ and $\sigma=-1+1/n$, then we have 
$d=r_0+2c/a_0H$ by \eqref{Proof-Lem-23-2000}.
Especially, we have $d\ge 0$ if $r_0\ge -2c/a_0H$,
and $d\le 0$ if $r_0\le -2c/a_0H$.

(v) 
If $H<0$ and $-1+1/n<\sigma<-1+2/n$, then we have 
$d(t)\nearrow \infty$ as $t\to T_0$ by \eqref{Proof-Lem-23-2000}.
Especially, we have $d\ge r_0+2c/a_0H\ge0$ if $r_0\ge -2c/a_0H$.

(vi) 
If $H<0$ and $\sigma=-1+2/n$, then we have 
\[
d=r_0+\frac{c\log(1+Ht)}{a_0H}+\frac{2c}{a_0H}
\nearrow \infty
\]
as $t\to T_0$.
Since $d(0)=r_0+2c/a_0H$, we have $d\ge 0$ if $r_0\ge -2c/a_0H$.

(vii) 
If $H<0$ and $\sigma>-1+2/n$, then we have 
\[
r_0+\frac{2c}{a_0H}=d(0)\le 
d\nearrow r_0-\frac{2c}{ a_0H\{n(1+\sigma)-2\} }
\]
as $t\to T_0$ by \eqref{Proof-Lem-23-2000}.
Thus, we have $d\ge 0$ if $r_0\ge -2c/a_0H$.

The inequality $\dot{q}\ge0$ holds if and only if $Hd\ge0$ by \eqref{Proof-Lem-23-1000}. 
So that, we obtain the result (2) by (i), (iii), (iv),
and the result (3) by (iv), (v), (vi), (vii).
\end{proof}


\newsection{Proof of Theorem \ref{Thm-22}}
\label{Sec-Thm-22}
We prove Theorem \ref{Thm-22} in this section.

(1) 
Taking the real part of the differential equation in \eqref{Cauchy}, 
and integrating by spatial variables, 
we have 
\[
c^{-2} \ddot{w}+M^2w-h=0,
\]
where $w(t):=\Re\int_\brn u(t,x)dx$ and 
$h(t):=\lambda a^{-n(p-1)/2}\int_\brn |u(t,x)|^pdx$.
We have 
\[
|w|^p\le (\omega_n r^n)^{p-1} \int_\brn |u(t,x)|^pdx
\]
by the finite speed of propagation of waves 
(see Proposition \ref{Prop-Finite}, below), 
by which 
we obtain 
$h\ge b|w|^p$, where 
\[
b:=\lambda 
\left(
\omega_n^{2/n} a r^2
\right)^{-n(p-1)/2}>0
\]
since $\lambda>0$, and $a$, $r$ are positive valued.
Thus, we obtain 
\[
c^{-2} \ddot{w}+M^2w-b|w|^p\ge 0,
\]
and 
the results in Lemma \ref{Lem-21} hold 
under the assumptions there,
which are assumed in Theorem \ref{Thm-22}.

We assume $b$ is a monotone function, and put 
\beq
\label{Def-TilB}
\widetilde{b}:=
\begin{cases}
b_0 & \mbox{if}\ \dot{b}\ge0,\\
b & \mbox{if}\ \dot{b}<0.
\end{cases}
\eeq
By $\widetilde{b}\le b$ and (3) in Lemma \ref{Lem-21},
we have 
\beq
\label{Proof-Thm-22-Diff-w}
0\le c^{-2}\ddot{w}-N^2w-\theta\widetilde{b} w^p
\le c^{-2} \ddot{w}-\theta\widetilde{b} w^p
\eeq
and $\dot{w}\ge w_1>0$ by (4) in Lemma \ref{Lem-21} 
and the second assumption in \eqref{Assume-w1}. 
By the multiplication of $\dot{w}$ 
in \eqref{Proof-Thm-22-Diff-w}, 
we have 
\[
\dot{E}+\frac{\theta \dot{\widetilde{b}} w^{p+1} }{p+1}\ge0,
\]
where 
\[
E(t):=\frac{\dot{w}^2(t)}{2c^2}-\frac{\theta \widetilde{b}(t) w^{p+1}(t)}{p+1}.
\]
Since we have $\dot{\widetilde{b} }\le 0$ by the definition 
\eqref{Def-TilB}, we obtain 
$\dot{E}\ge0$ and thus, $E(t)\ge E(0)$ for $0\le t< T_0$.
So that, under 
\beq
\label{Proof-Thm-22-10}
E(0)\ge0, 
\eeq
we have $E\ge0$, which yields 
\[
\dot{w}^2
\ge
\frac{2c^2\theta \widetilde{b} w^{p+1} }{p+1}
\ge 
\frac{2c^2\theta w_0^{(1-\varepsilon)(p-1)} }{p+1}
\widetilde{b} e^{cN(1-\varepsilon)(p-1)t} w^{2+\varepsilon(p-1)}
\]
by (1) in Lemma \ref{Lem-21}.
Therefore, we obtain 
\beq
\label{ODE-w}
\dot{w}\ge C w^\alpha,
\eeq
where $\alpha:=1+\varepsilon(p-1)/2$ and 
a nonnegative constant $C$ is defined by 
\beq
\label{Def-C}
C^2:=
\frac{2c^2\theta w_0^{(1-\varepsilon)(p-1)} }{p+1}
\inf_{0<t<T_0} \left\{\widetilde{b} e^{cN(1-\varepsilon)(p-1)t} \right\}.
\eeq
Multiplying $w^{-\alpha}$ to the both sides of \eqref{ODE-w}, 
we have $d(w^{1-\alpha})/dt\le C(1-\alpha)$.
Integrating the both sides, we obtain
\[
w(t)\ge w_0
\left\{1-C(\alpha-1)w_0^{\alpha-1} t\right\}^{-1/(\alpha-1)},
\]
by which $w$ blows up no later than  
$t=1/\left\{C(\alpha-1)w_0^{\alpha-1} \right\}=T_\ast$ 
defined in \eqref{Def-D-T} under 
\beq
\label{Proof-Thm-22-12}
C^2>0.
\eeq

Putting $Q:=\omega_n^{2/n} a_0$, and using $q$ and $\widetilde{q}$ in \eqref{Def-r-q} and \eqref{Def-TildeQ},  
we have 
\beq
\label{Proof-Thm-22-b}
b=\frac{\lambda}{ (Qq)^{n(p-1)/2} }
\ \ \mbox{and}\ \ 
\widetilde{b}=\frac{\lambda}{ (Q\widetilde{q})^{n(p-1)/2} },
\eeq
by which $b$ is monotone if and only if $q$ is monotone,
and $\dot{b}>0$ is equivalent to $\dot{q}<0$ by $\lambda>0$.
The first assumption in \eqref{w0} in Lemma \ref{Lem-21} is satisfied if 
\[
w_0>\sup_{0<t<T_0}
\left[
e^{-cNt}
\left\{
\frac{N^2+M^2(t)}{(1-\theta)\widetilde{b}(t)} 
\right\}^{1/(p-1)}
\right]
\]
by $b\ge \widetilde{b}$, 
which is satisfied by \eqref{Assume-w0}
since we have 
\beq
\label{b-q}
\widetilde{b} e^{cN(p-1)t}=
\lambda Q^{-n(p-1)/2} 
\left(
\widetilde{q}^{-n/2} e^{cNt} \right)^{p-1}.
\eeq

We check the condition $E(0)\ge0$ in \eqref{Proof-Thm-22-10}. 
This is rewritten as 
\[
w_1^2
\ge 
\frac{2c^2\theta b_0w_0^{p+1} }{p+1}
=
\frac{2\lambda c^2\theta w_0^{p+1} }{ (p+1)(q_0 Q)^{n(p-1)/2} },
\]
which is the second assumption in \eqref{Assume-w1} 
by 
$b_0=\lambda(q_0Q)^{-n(p-1)/2}$ 
and 
$q_0=r_0^2$ 
due to \eqref{Proof-Thm-22-b} and \eqref{Def-r-q}.

We check the condition $C^2>0$ in \eqref{Proof-Thm-22-12}.
Since $C^2$ is rewritten as 
\[
C^2
=
\frac{2\lambda c^2\theta}{p+1}
\left\{
\frac{w_0^{1-\varepsilon}}{Q^{n/2}}
\inf_{0<t<T_0}
\frac{e^{cN(1-\varepsilon)t} }{\widetilde{q}^{n/2} }
\right\}^{p-1}
\]
by \eqref{b-q}, 
the condition $C^2>0$ holds by the first assumption in \eqref{Assume-AB}.


\newsection{Proof of Corollary \ref{Cor-24}}
\label{Sec-Cor-24}
We exclude the case $(1+\sigma)H<0$ with $\sigma<0$, 
under which $M^2$ tends to $-\infty$ as $t\to T_0$ by (vi) in (3) in Lemma \ref{Lem-11}.
We note that this case $(1+\sigma)H<0$ with $\sigma<0$ holds if and only if 
$H>0$ with $\sigma<-1$, or $H<0$ with $-1<\sigma<0$.
So that, we consider the following cases (I), (I\hspace{0mm}I) and (I\hspace{0mm}I\hspace{0mm}I).
\beq
\label{Cor-24-1000}
\begin{cases}
(I) & H=0,\ \sigma\in \br. \\
(I\hspace{-1mm}I) & H>0,\ \sigma\ge-1. \\
(I\hspace{-1mm}I\hspace{-1mm}I) & H<0,\ \sigma\ge0 \ \mbox{or}\ \sigma\le -1.
\end{cases}
\eeq
The second assumption in \eqref{Cond-NM} is satisfied by the conditions from (i) to (viii) in Corollary \ref{Cor-24}  
by (3) in Lemma \ref{Lem-11}.

The monotonicity assumption on $q$ in Theorem \ref{Thm-22} is satisfied 
if one of the following conditions holds by Lemma \ref{Lem-23}.
\beq
\label{Proof-Cor-24-1000}
\begin{cases}
(\four) & H\ge0, \ \sigma\in \br, \ r_0>0. \\ 
(V) & H<0, \ \sigma\le-1+\frac{1}{n}, \ r_0\le -\frac{2c}{a_0H}. \\
(\six) & H<0, \ \sigma\ge-1+\frac{1}{n}, \ r_0\ge -\frac{2c}{a_0H}. 
\end{cases}
\eeq

We note that $a$ defined by $\eqref{Def-a}$ has an exponential order if $H\neq0$ and $\sigma=-1$,  
and a polynomial order otherwise.
We note that $M^2$ defined by \eqref{Def-M} has a polynomial order 
by (2) in Lemma \ref{Lem-11}.
We also note that $r$ defined by \eqref{Def-r-q} has a polynomial order if $\sigma\neq-1$,
while it has exponential order if $H\neq0$ and $\sigma=-1$.
On $q$ defined by \eqref{Def-r-q},
if $H<0$ and $\sigma=-1$, then 
\begin{eqnarray}
q
&=&
e^{Ht}\left\{r_0+\frac{c}{a_0H}\left(1-e^{-Ht}\right)\right\}^2 \nonumber\\
&=&
e^{-Ht}\left\{-\frac{c}{a_0H}+\left(r_0+\frac{c}{a_0H}\right) e^{Ht}\right\}^2 \nonumber\\
&\simeq& 
\left(\frac{c}{a_0H}\right)^2 e^{-Ht}.
\label{Proof-Cor-24-q}
\end{eqnarray}
Similarly, or more easily, we also have 
$q\simeq r_0^2 e^{Ht}$ if $H\ge0$ and $\sigma=-1$, 
and $q$ has a polynomial order if $H\in \br$ and $\sigma\neq-1$.

Since $\dot{q}\ge0$ if $\sigma=-1$ by (2) in Lemma \ref{Lem-23} 
under \eqref{Proof-Lem-23-1000}, 
we have $\widetilde{q}=q\simeq e^{|H|t}$ by 
\eqref{Def-TildeQ} and 
\eqref{Proof-Cor-24-q},
while $\widetilde{q}$ has a polynomial order if $\sigma\neq-1$.
Thus, we have 
\[
\frac{e^{cN(1-\varepsilon)t}}{\widetilde{q}^{n/2}}
\simeq
\begin{cases}
e^{\left\{cN(1-\varepsilon)-n|H|/2\right\}t} & \mbox{if} \ \sigma=-1,\\
P(t) e^{cN(1-\varepsilon)t} & \mbox{if} \ \sigma\neq-1,
\end{cases}
\]
where $P(t)$ is a positive function with a polynomial order, 
which yields 
\[
\inf_{0<t<T_0}
\left\{
\frac{ e^{cN(1-\varepsilon)t} }{\widetilde{q}^{n/2}(t) }
\right\}
>0
\]
if one of the following conditions holds.
\beq
\label{Proof-Cor-24-2000}
\begin{cases}
(\seven) & N>\frac{n|H|}{2c(1-\varepsilon)},\ \sigma=-1,
\\ 
(\eight) & N>0,\ \sigma\neq-1.
\end{cases}
\eeq 
Namely, the first assumption in \eqref{Assume-AB} is satisfied.

Similarly, we have 
\begin{eqnarray*}
\frac{ \widetilde{q}^{n/2}(t) \left\{N^2+M^2(t)\right\}^{1/(p-1)} }{ e^{cNt} }
&\simeq&
\begin{cases}
P(t) e^{(n|H|/2-cN)t} & \mbox{if} \ \sigma=-1,\\
P(t) e^{-cNt} & \mbox{if} \ \sigma\neq-1,
\end{cases}
\\
&\lesssim& 1,
\end{eqnarray*}
which yields 
\[
\sup_{0<t<T_0}
\frac{ \widetilde{q}^{n/2}(t) \left\{N^2+M^2(t)\right\}^{1/(p-1)} }{ e^{cNt} }
<\infty
\]
if one of the following conditions in \eqref{Proof-Cor-24-2000} holds.
Namely, the second assumption in \eqref{Assume-AB} is satisfied.

Therefore, the assumptions in Theorem \ref{Thm-22} are all satisfied 
under the conditions from (I) to 
(V\hspace{-0.5mm}I\hspace{-0.5mm}I\hspace{-0.5mm}I) 
in 
\eqref{Cor-24-1000},
\eqref{Proof-Cor-24-1000} 
\eqref{Proof-Cor-24-2000}.
These conditions from (I) to 
(V\hspace{-0.5mm}I\hspace{-0.5mm}I\hspace{-0.5mm}I) 
are satisfied by the conditions from (i) to (viii) in Corollary \ref{Cor-24}.
Indeed, 
the assumption (i) satisfies (I), (I\hspace{-0.5mm}V), 
(V\hspace{-0.5mm}I\hspace{-0.5mm}I), 
(V\hspace{-0.5mm}I\hspace{-0.5mm}I\hspace{-0.5mm}I),
the assumptions (ii) and (iii) satisfy (II), (I\hspace{-0.5mm}V), (V\hspace{-0.5mm}I\hspace{-0.5mm}I\hspace{-0.5mm}I),
the assumption (iv) satisfies (I\hspace{-0.5mm}I), (I\hspace{-0.5mm}V), (V\hspace{-0.5mm}I\hspace{-0.5mm}I),
the assumption (v) satisfies (I\hspace{-0.5mm}I\hspace{-0.5mm}I), (V\hspace{-0.5mm}I), (V\hspace{-0.5mm}I\hspace{-0.5mm}I\hspace{-0.5mm}I),
the assumption (vi) satisfies 
(I\hspace{-0.5mm}I\hspace{-0.5mm}I), 
(V), 
(V\hspace{-0.5mm}I), 
(V\hspace{-0.5mm}I\hspace{-0.5mm}I\hspace{-0.5mm}I),
the assumption (vii) satisfies (I\hspace{-0.5mm}I\hspace{-0.5mm}I), (V), (V\hspace{-0.5mm}I\hspace{-0.5mm}I),
the assumption (viii) satisfies (I\hspace{-0.5mm}I\hspace{-0.5mm}I), (V), (V\hspace{-0.5mm}I\hspace{-0.5mm}I\hspace{-0.5mm}I).

\newsection{Appendix: Finite speed of propagation}
\label{Sec-Appendix}
In this section, we recall a fundamental result on the finite speed of propagation of complex-valued waves in FLRW spacetimes for the completeness of the paper 
although it is a straight extension of the argument 
in \cite[Thorem 2.2]{Sogge-2008-IntPress} for real-valued waves in the Minkowski spacetime.

Let $n\ge1$, $c>0$ and $0<t_0<T_0\le \infty$.
Let $a\in C^1([0,T_0),(0,\infty))$ be a monotone function.
Put 
\[
r_0:=\int_0^{t_0}\frac{c}{a(\tau)} d\tau, 
\ \ 
r(t):=\int_t^{t_0}\frac{c}{a(\tau)} d\tau
=r_0-\int_0^{t}\frac{c}{a(\tau)} d\tau
\]
for $0\le t\le t_0$, 
where we consider the backward cone given by $r(t)$ which is different from the forward cone $r(t)$ in \eqref{Def-r-q}.
Put $B:=\{x\in \brn; \ |x|\le r_0\}$. 
Let $\psi$ be the inverse mapping of $r(t)$, i.e., $t=\psi(r)$.
We consider the inside of the backward cone $D$ defined by 
\begin{eqnarray*}
D
&:=&\{(t,x)\in [0,t_0]\times \brn;\ x\in B,\ 0\le t\le \psi(|x|)\}
\\
&=&\{(t,x)\in [0,t_0]\times \brn;\ 0\le t\le t_0, 0\le r\le r(t)\}.
\end{eqnarray*}
Let $h\in C^1([0,t_0]\times\brn\times \bc\times\bc^{1+n}\times \bc^{(1+n)^2},\bc)$,
where we regard derivatives by complex variables as those by real variables of their real and imaginary parts.
Put 
\[
\nabla:=(\partial_1,\ldots, \partial_n),\ \ \nabla_g:=(c^{-1}\partial_t, a^{-1}\nabla).
\]
We consider the Cauchy problem given by 
\beq
\label{KG}
\begin{cases}
c^{-2}\partial_t^2 u(t,x)-a^{-2}(t) \Delta u(t,x)+h(t,x,u(t,x), \nabla_g u(t,x), \nabla_g^2u(t,x))=0,\\
u(0,x)=u_0(x),\ \ \partial_tu(0,x)=u_1(x)
\end{cases}
\eeq
for $(t,x)\in D$, 
where $\Delta:=\sum_{j=1}^n \partial_j^2$.
We show that $u$ vanishes in $D$ if $u_0=u_1=0$ on $B$ as follows.

\begin{lemma}
\label{Appendix-Prop}
The solution $u\in C^2(D,\bc)$ of \eqref{KG} with $u_0=u_1=0$ on $B$ satisfies $u=0$ on $D$.
\end{lemma}

\begin{proof}
Let $s$ be a parameter with $0\le s\le t_0$. 
Put 
\[
\chi_s(r):=
\sqrt{
r(s)^2+\frac{\left\{r_0^2-r(s)^2\right\}r^2}{r_0^2}
},
\ \ 
\phi_s(r):=\psi\left(\chi_s(r)\right)
\]
for $0\le r\le r_0$.
Put
\begin{eqnarray*}
D_s
&:=&
\{(t,x)\in [0,t_0]\times \brn;\ x\in B,\ 0\le t\le \phi_s(|x|)\},
\\
\Lambda_s
&:=&
\left\{
(t,x)\in [0,t_0]\times \brn;\ x\in B,\ t=\phi_s(|x|)
\right\},
\end{eqnarray*}
where we note $D=D_{t_0}$,
and $\Lambda_s$ denotes the upper boundary of $D_s$.
We note $\chi_0(r)=\chi_s(r_0)=r_0$, $\chi_s(0)=r(s)$ and $\max\{r,r(s)\}\le \chi_s(r)\le r_0$ by the definition of $\chi$.
Thus, we have 
$\phi_s(r)=0$ if $s=0$ or $r=r_0$, $\phi_s(r)=s$ if $r=0$, and 
$0\le \phi_s(r)\le \min\{s,\psi(r)\}$ 
since $\psi$ is a monotone decreasing function.
We also have 
\beq
\label{Appendix-1000}
\left|\partial_s\chi_s(r)\right|\le \frac{c}{a(s)},
\ \  
\left|\partial_s\phi_s(r)\right|\le \frac{a(\phi_s(r))}{a(s)}
\eeq
for $0\le s\le t_0$ and $0\le r\le r_0$. 
Let $t_1$ be an arbitrarily fixed number with $0< t_1<t_0$. 
We have 
\beq
\label{Appendix-1200}
0\le \partial_r\chi_s(r)\le \theta(t_1)<1,
\ \ 
\left|\partial_r\phi_s(r)\right|\le \frac{a(\phi_s(r))\theta(t_1)}{c}
\eeq
for $0\le s\le t_1$ and $0\le r\le r_0$ 
by 
\[
\partial_r\chi_s(r)
=
\frac{\left\{r_0^2-r(s)^2\right\}r}{\chi_s(r)r_0^2}
\le
\frac{ \sqrt{r_0^2-r(t_1)^2} }{r_0}
=:
\theta(t_1),
\]
\[
\psi'(r)=-a(\psi(r))/c,
\]
 and 
\[
\partial_r\phi_s(r)=-\frac{a(\phi_s(r))\partial_r\chi_s(r)}{c},
\]
where $\theta(t_1)<1$ holds by $r(t_1)>0$.
We note that the estimate 
\beq
\label{Appendix-1500}
0\le \phi_{s+ds}(r)-\phi_s(r)\le 
\begin{cases}
ds & \mbox{if}\ \dot{a}\ge0,
\\
\frac{a_0 ds}{a(s+ds)} & \mbox{if}\ \dot{a}\le 0
\end{cases}
\eeq
for $0\le s\le s+ds\le t_0$  
holds by the second inequality in \eqref{Appendix-1000}.

(1) We consider the case $\dot{a}\ge0$.
Multiplying $\overline{\partial_t u}$ to the both sides of the differential equation in \eqref{KG}, 
and taking real parts, we have the divergence form 
\[
\partial_t e^0+\sum_{j=1}^n \partial_j e^j+e^{n+1}+2\Re\left(\overline{\partial_tu} h\right)=0,
\]
where 
$e^0:=c^{-2}|\partial_tu|^2+a^{-2}|\nabla u|^2$,
$(e^1,\ldots,e^n):=-2a^{-2}\Re\left(\overline{\partial_tu} \nabla u\right)$,
$e^{n+1}:=2a^{-3}\dot{a}|\nabla u|^2$.
Let $0\le s\le t_1<t_0<T_0$.
Integrating the both sides of this form on $D_s$, 
we have
\[
I_s+\iint_{D_s} e^{n+1} dxdt+2\Re\iint_{D_s} \overline{\partial_t u} h dx dt=0,
\]
where  
\[
I_s:=\iint_{D_s} \partial_t e^0+\sum_{j=1}^n \partial_j e^j dx dt,
\]
which and the assumption $\dot{a}\ge0$ yield $e^{n+1}\ge0$ and 
\beq
\label{Appendix-2000}
I_s\le 2\iint_{D_s} |\partial_t u h| dx dt.
\eeq

We estimates the both sides in \eqref{Appendix-2000} separately.
By the divergence theorem, we have
\beq
\label{Appendix-3000}
I_s=\int_{\Lambda_s} (e^0,\ldots,e^n)\cdot \nu d\sigma 
\ge 
\left\{1-\theta(t_1)\right\}J_s,
\eeq
where $d\sigma$ denotes the measure on $\Lambda_s$, 
and $\nu$ denotes the unit outer normal vector on $\Lambda_s$, 
we have put 
\[
r=|x|,
\ \ 
\nu:=
\frac{(1,-x\partial_r\phi_s(r)/r)}{ \sqrt{1+|\partial_r\phi_s(r)|^2} },
\ \ 
J_s:=\int_{\Lambda_s} \frac{e^0}{\sqrt{1+|\partial_r\phi_s(r)|^2}} d\sigma
\]
and we have used $u_0=u_1=0$ on $B$ and 
\[
(e^0,\ldots,e^n)\cdot \nu
\ge 
\frac{(1-\theta(t_1))e^0}{\sqrt{1+|\partial_r\phi_s(r)|^2}}.
\]

On the other hand, since we have 
\[
|h(t,x,u,\nabla_gu,\nabla_g^2u)|\lesssim N(|u|+|\nabla_g u|)
\]
by $h(t,x,0,0,\nabla_g^2u)=0$, we have 
\beq
\label{Appendix-4000}
|\partial_t u h(t,x,u,\nabla_g u,\nabla_g^2u)|\lesssim cN(e^0+|u|^2)
\eeq
for $(t,x)\in D$, 
where we have put 
\begin{eqnarray*}
M&:=&\max_{(t,x)\in D}
\left\{
|u(t,x)|,
|\nabla_g u(t,x)|,
|\nabla_g^2 u(t,x)|
\right\},
\\
N&:=&\max_{\substack{(t,x)\in D \\ |y|,|z|,|w|\le M} } 
\left|
\nabla_{y,z}h(t,x, y,z,w)
\right|.
\end{eqnarray*}
By $u_0=0$ and $u(t,\cdot)=\int_0^t\partial_t u(\tau,\cdot)d\tau$ on $D$, 
we have 
\beq
\label{Appendix-4500}
\iint_{D_s}|u|^2dx dt
=\int_B \int_0^{\phi_s(|x|)} |u(t,x)|^2 dt dx \le \frac{c^2t_1^2}{2} \iint_{D_s} e^0 dx d\tau.
\eeq
So that, we have 
\beq
\label{Appendix-5000}
\iint_{D_s} |\partial_tu h| dx dt
\lesssim
c(1+c^2t_1^2)N\iint_{D_s} e^0 dx d\tau
\lesssim
c(1+c^2t_1^2)N\int_0^sJ_\tau d\tau,
\eeq
where we have used \eqref{Appendix-1500} for the last inequality.

By \eqref{Appendix-2000}, \eqref{Appendix-3000} and \eqref{Appendix-5000}, 
we obtain 
\[
J_s\lesssim \frac{c(1+c^2t_1^2)N}{1-\theta(t_1)} \int_0^sJ_\tau d\tau
\]
for $0\le s\le t_1$, 
which shows $J_s=0$ for $0\le s\le t_1$ by the Gronwall inequality.
Therefore, $u=0$ holds on $D_{t_1}$ by $u_0=0$ on $B$.
Since $t_1$ is arbitrary for $0< t_1<t_0$, we obtain $u=0$ on $D$.

(2) 
Next, we consider the case $\dot{a}\le0$. 
Since the proof is similar to the case $\dot{a}\ge0$, 
we show its outline.
Multiplying $a^2\overline{\partial_t u}$ to the both sides of the differential equation in \eqref{KG}, 
we have  
\[
\partial_t f^0+\sum_{j=1}^n \partial_j f^j+f^{n+1}+2a^2\Re\left(\overline{\partial_tu} h\right)=0,
\]
where 
$f^0:=c^{-2}a^2|\partial_tu|^2+|\nabla u|^2$,
$(f^1,\ldots,f^n):=-2\Re\left(\overline{\partial_tu} \nabla u\right)$,
$f^{n+1}:=-2c^{-2}a\dot{a}|\partial_t u|^2$. 
Integrating the both sides of this form on $D_s$, 
we have
\[
K_s+\iint_{D_s} f^{n+1} dxdt+2\Re\iint_{D_s} a^2\overline{\partial_t u} h dx dt=0,
\]
where  
\[
K_s:=\iint_{D_s} \partial_t f^0+\sum_{j=1}^n \partial_j f^j dx dt,
\]
which and the assumption $\dot{a}\le0$ yield $f^{n+1}\ge0$ and 
\beq
\label{Appendix-2000-2}
K_s\le 2\iint_{D_s} a^2|\partial_tu h| dx dt.
\eeq

We estimates the both sides in \eqref{Appendix-2000-2} separately.
We have
\beq
\label{Appendix-3000-2}
K_s=\int_{\Lambda_s} (f^0,\ldots,f^n)\cdot \nu d\sigma 
\ge 
\left\{1-\theta(t_1)\right\}L_s,
\eeq
where 
\[
L_s:=\int_{\Lambda_s} \frac{f^0}{\sqrt{1+|\partial_r\phi_s(r)|^2}} d\sigma
\]
and we have used 
\[
(f^0,\ldots,f^n)\cdot \nu
\ge 
\frac{(1-\theta(t_1))f^0}{\sqrt{1+|\partial_r\phi_s(r)|^2}}.
\]

On the other hand, we have 
\[
a^2|\partial_t u h(t,x,u,\nabla_gu,\nabla_g^2u)|\lesssim cN(f^0+a^2|u|^2),
\]
similarly to \eqref{Appendix-4000}.
Since we have 
\[
a^2(t)|u(t,\cdot)|^2
\le t a^2(t)\int_0^t|\partial_t u(\tau,\cdot)|^2 d\tau
\le t \int_0^ta^2(\tau)|\partial_t u(\tau,\cdot)|^2 d\tau
\]
by the monotonicity $a(t)\le a(\tau)$ for $0\le \tau\le t$,
we obtain 
\[
\int_0^{\phi_s(r)}a^2(t)|u(t,\cdot)|^2 dt
\le 
\frac{s^2}{2}
\int_0^{\phi_s(r)}a^2(\tau)|\partial_t u(\tau,\cdot)|^2 d\tau.
\]
Thus, we have 
\begin{eqnarray*}
\iint_{D_s}a^2|u|^2dx dt
&=&
\int_B  \int_0^{\phi_s(|x|)} a^2(t)|u(t,x)|^2 dt dx 
\\
&\le&
\frac{s^2}{2}
\iint_{D_s} a^2(\tau)|\partial_t u(\tau,x)|^2 dx d\tau
\\
&\le& 
\frac{c^2t_1^2}{2} \iint_{D_s} f^0 dxd\tau 
\end{eqnarray*}
similarly to \eqref{Appendix-4500}.
So that, we have 
\beq
\label{Appendix-5000-2}
\iint_{D_s} a^2|\partial_tu h| dxdt
\lesssim
c(1+c^2t_1^2)N\iint_{D_s} f^0 dxd\tau
\lesssim
\frac{c(1+c^2t_1^2)Na_0}{a(t_1)}\int_0^s L_\tau d\tau,
\eeq
where we have used \eqref{Appendix-1500} for the last inequality.

By \eqref{Appendix-2000-2}, \eqref{Appendix-3000-2} and \eqref{Appendix-5000-2}, 
we obtain 
\[
L_s\le \frac{c(1+c^2t_1^2)Na_0}{\left\{1-\theta(t_1)\right\}a(t_1)} \int_0^sL_\tau d\tau
\]
for $0\le s\le t_1$.
We obtain $u=0$ on $D$ similarly to the case $\dot{a}\le 0$. 
\end{proof}

We have the following result immediately from Lemma \ref{Appendix-Prop},
which shows the finite speed of propagation of \eqref{KG}.

\begin{proposition}
\label{Prop-Finite}
Let $0<T\le T_0$, $x_0\in \brn$ and $R>0$.
Put $B_R(x_0):=\{x\in \brn;\ |x-x_0|\le R\}$ and 
\[
D_R(T,x_0):=\left\{(t,x)\in [0,T)\times\brn;\ |x-x_0|\le R+\int_0^t \frac{c}{a(\tau)} d\tau \right\}.
\]
The solution $u\in C^2([0,T)\times \brn,\bc)$ of \eqref{KG} with 
$\supp u_0\cup \supp u_1\subset B_R(x_0)$ 
satisfies $u=0$ on 
$[0,T)\times\brn\backslash D_R(T,x_0)$.
\end{proposition}

\begin{proof}
For any $x_1\in \brn$ outside of $B_R(x_0)$, 
we have $u_0=u_1=0$ on $B_{\delta}(x_1)$ for sufficiently small $\delta>0$.
We have $u=0$ inside of the backward cone whose bottom is $B_{\delta}(x_1)$ by Lemma \ref{Appendix-Prop}.
Since $x_1$ is arbitrary, we have $u=0$ outside of $D_R(T,x_0)$.
\end{proof}

\vspace{10pt}

%

{\bf Acknowledgments.}
This work was supported by JSPS KAKENHI Grant Numbers 16H03940, 17KK0082, 22K18671.


%


\begin{thebibliography}{999}

\bibitem{Balogh-Banda-Yagdjian-2019-CommNonlinearSciNumerSimulat} 
A. Balogh, J. Banda, K. Yagdjian, 
\emph{High-performance implementation of a Runge-Kutta finite-difference scheme for the Higgs boson equation in the de Sitter spacetime}, 
Commun. Nonlinear Sci. Numer. Simul. {\bf 68} (2019), 15--30. 

\bibitem{Baskin-2012-AHP}
D. Baskin,
\emph{Strichartz estimates on asymptotically de Sitter spaces},
Annales Henri Poincar{\'e} {\bf 14} (2013), Issue 2, pp 221--252.

\bibitem{Carroll-2004-Addison}
S. Carroll, 
\emph{Spacetime and geometry.  
An introduction to general relativity}, 
Addison Wesley, San Francisco, CA, 2004, xiv+513 pp.

\bibitem{DInverno-1992-Oxford}
R. d'Inverno, 
\emph{Introducing Einstein's relativity}, 
The Clarendon Press, Oxford University Press, New York, 1992, xii+383 pp. 

\bibitem{Galstian-Yagdjian-2015-NA-TMA} 
A. Galstian, K. Yagdjian, 
\emph{Global solutions for semilinear Klein-Gordon equations in FLRW spacetimes},
Nonlinear Anal. {\bf 113}(2015), 339--356.

\bibitem{Galstian-Yagdjian-2020-RMP}
A. Galstian, K. Yagdjian, 
\emph{Finite lifespan of solutions of the semilinear wave equation in the Einstein-de Sitter spacetime}, 
Rev. Math. Phys. {\bf 32} (2020), no. 7, 2050018, 31 pp.

\bibitem{Hintz-Vasy-2015-AnalysisPDE}
P. Hintz, A. Vasy,  
\emph{Semilinear wave equations on asymptotically de Sitter, Kerr-de Sitter and Minkowski spacetimes}, 
Anal. PDE  {\bf 8} (2015),  no. 8, 1807--1890.

\bibitem{Kato-1980-CPAM}
T. Kato, 
\emph{Blow-up of solutions of some nonlinear hyperbolic equations}, 
Comm. Pure Appl. Math.  {\bf 33} (1980),  no. 4, 501--505. 

\bibitem{McCollum-Mwamba-Oliver-2024-NA}
J. McCollum, G. Mwamba, J. Oliver, 
\emph{A sufficient condition for blowup of the nonlinear Klein-Gordon equation with positive initial energy in FLRW spacetimes}, 
Nonlinear Anal. {\bf 246} (2024), Paper No. 113582, 11 pp.

\bibitem{Nakamura-2014-JMAA}
M. Nakamura, 
\emph{The Cauchy problem for semi-linear Klein-Gordon equations in de Sitter spacetime},
J. Math. Anal. Appl.  {\bf 410}  (2014),  no. 1, 445--454. 

\bibitem{Nakamura-2021-JMP}
M. Nakamura, 
\emph{The Cauchy problem for the Klein-Gordon equation under the quartic potential in the de Sitter spacetime}, 
J. Math. Phys. {\bf 62} (2021), no. 12, Paper No. 121509, 21 pp.

\bibitem{Nakamura-Yoshizumi-9999}
\label{Nakamura-Yoshizumi-9999}
M. Nakamura, T. Yoshizumi,
\emph{The Cauchy problem for semi-linear Klein-Gordon equations in Friedmann-Lema\^itre-Robertson-Walker spacetimes}, preprint.

\bibitem{Sogge-2008-IntPress}
C. D. Sogge, 
\emph{Lectures on non-linear wave equations}, 
Second edition,
International Press, Boston, MA, 2008, x+205 pp.

\bibitem{Wei-Yong-2024-JMP}
C. Wei, Z. Yong, 
\emph{Global existence and blowup of smooth solutions to the semilinear wave equations in FLRW spacetime}, 
J. Math. Phys. {\bf 65} (2024), no.5, Paper No. 051504, 21 pp.

\bibitem{Yagdjian-2009-DCDS}
K. Yagdjian, 
\emph{The semilinear Klein-Gordon equation in de Sitter spacetime},
Discrete Contin. Dyn. Syst. Ser. S  {\bf 2} (2009),  no. 3, 679--696. 

\bibitem{Yagdjian-2012-JMAA}
K. Yagdjian,
\emph{Global existence of the scalar field in de Sitter spacetime},
J. Math. Anal. Appl. {\bf 396} (2012), no. 1, 323--344. 


\end{thebibliography}
\end{document}